\documentclass[11pt]{article}
\usepackage{curves}
\usepackage{url}
\usepackage[dvips]{graphicx}
\usepackage{amssymb,amsfonts,amsmath,amsthm,amscd,dsfont,mathrsfs,bbold}
\usepackage{graphicx,float,psfrag,epsfig,color}

\newcommand\independent{\protect\mathpalette{\protect\independent}{\perp}} 
\def\independent#1#2{\mathrel{\rlap{$#1#2$}\mkern2mu{#1#2}}}

\newcommand{\iid}{\stackrel{\text{iid}}{\sim}}

\newcommand{\F}{\mathbb{F}}

\newcommand{\mZ}{\mathbb{Z}}

\newcommand{\pp}{\mathbb{P}}
\newcommand{\E}{\mathbb{E}}

\newcommand{\e}{\varepsilon}

\newcommand{\rank}{\mathrm{rank}}
\newcommand{\nullity}{\mathrm{nullity}}

\newcommand{\X}{\mathcal{X}}
\newcommand{\Y}{\mathcal{Y}}

\theoremstyle{definition}
\newtheorem{definition}{Definition}%[section]
\theoremstyle{plain}
\newtheorem{thm}{Theorem}%[section]
\theoremstyle{plain}
\newtheorem{prop}{Proposition}%[section]
\theoremstyle{plain}
\newtheorem{lemma}{Lemma}%[section]
\theoremstyle{plain}
\newtheorem{corol}{Corollary}%[thm]
\theoremstyle{plain}
\theoremstyle{remark}
\newtheorem{remark}{Remark}%[section]
\theoremstyle{discussion}
\theoremstyle{plain}
\begin{document}
%check reverse implication for source-channel, 

\title{
Randomness and dependencies extraction via polarization, with applications to Slepian-Wolf coding and secrecy
%Polarization of correlated sources, distributed compression and secrecy
%Slepian-Wolf coding and secret key generation using polar coding
}

\author{Emmanuel Abbe\thanks{This paper was presented at ITA 2010 \cite{corr}, while the secrecy application was presented at the Allerton Conference on Communication, Control and Computing, Monticello, October 2011. The author is now with the Electrical Engineering Department and Program in Applied and Computational Mathematics at Princeton University, Princeton, USA.}\\
Ecole Polytechnique F\'ed\'erale de Lausanne 
}
%  \IEEEauthorblockN{Emmanuel Abbe}
%  \IEEEauthorblockA{ Ecole Polytechnique F\'ed\'erale de Lausanne \\
%    Email: emmanuel.abbe@epfl.ch}
%%\and
%%  \IEEEauthorblockN{Joachim Rosenthal}
%%  \IEEEauthorblockA{Institute of Mathematics\\
%%    University of Zurich\\
%%    Winterthurerstr. 190, 8057 Zurich, Switzerland\\
%%    Email: rosenthal@math.uzh.ch}
%}
% Conference papers typically do not use \thanks and this command is
% locked out in conference mode.  To credit grants etc. use an
% acknowledgments section at the end of the paper.

% For over three affiliations, or if they do not fit all within the width
% \author{\IEEEauthorblockN{Marcus Greferath\IEEEauthorrefmark{1},
%     Jens Zumbr\"agel\IEEEauthorrefmark{1},
%     Joachim Rosenthal\IEEEauthorrefmark{2},
%     Jean-Luc Picard\IEEEauthorrefmark{3},
%     and Lisa Simpson\IEEEauthorrefmark{4}}
%   \IEEEauthorblockA{\IEEEauthorrefmark{1}%
%     Claude Shannon Institute, University College Dublin, 
%     Belfield, Dublin 4, Ireland\\
%     Email: \{marcus.greferath, jens.zumbragel\}@ucd.ie}
%   \IEEEauthorblockA{\IEEEauthorrefmark{2}%
%     Institute of Mathematics, University of Zurich,
%     Winterthurerstr. 190, 8057 Zurich, Switzerland\\
%     Email: rosenthal@math.uzh.ch}
%   \IEEEauthorblockA{\IEEEauthorrefmark{3}%
%     Starfleet Academy, San Francisco, California 96678-2391\\
%     Telephone: (800) 555--1212, Fax: (888) 555--1212}
%   \IEEEauthorblockA{\IEEEauthorrefmark{4}%
%     Twentieth Century Fox, Springfield, USA\\
%     Email: lisa@thesimpsons.com}}

% Make the title area.

\date{}
\maketitle

\begin{abstract}
  \boldmath
The polarization phenomenon for a single source is extended to a framework with multiple correlated sources.
It is shown in addition to extracting the randomness of the source, the polar transforms takes the original arbitrary dependencies to extremal dependencies.   
Polar coding schemes for the Slepian-Wolf problem and for secret key generations are then proposed based on this phenomenon. In particular, constructions of secret keys achieving the secrecy capacity and compression schemes achieving the Slepian-Wolf capacity region are obtained with a complexity of $O(n \log (n))$.

%This result is then used to: (1) provide a low complexity and sum-rate achieving polar coding scheme for distributed data compression, i.e., Slepian-Wolf coding (without decomposing the problem into single-user problems), (2) provide a low complexity and optimal (achieving entropy) polar coding scheme for sources with memory, (3) show a generalized polarization phenomenon for arbitrary finite fields (generalizing the result for fields of prime cardinality).

\end{abstract}

% ------------------------------------------------------------

\section{Introduction}

The polarization technique has been developed in \cite{ari} for channel coding and then in \cite{ari3} for source coding. 
The codes resulting from this technique, called polar codes, have several desirable attributes: (1) they have low encoding and decoding complexity (2) bounds on the error probability which are exponential in the square root of the block length are obtained (3) they allow to achieve the Shannon capacity on symmetric discrete memoryless channels.

One of the key results in the development of polar codes is the following `polarization phenomenon'. 
For a matrix $A$, we denote by $A^{\otimes k}$ the matrix obtained by taking $k$ Kronecker products of $A$ with itself. 
\begin{thm}\label{thmari}\cite{ari,ari3}
%Let $X=[X_1,\dots,X_n]$ be i.i.d. random variables with distribution $p$ on $\F_q$, $n=2^\ell$, and let $U^n=X^n G_n$, where $G_n=    \bigl[\begin{smallmatrix} % or pmatrix or bmatrix or Bmatrix or ...
%      1 & 0 \\
%      1 & 1 \\
%   \end{smallmatrix}\bigr]^{\otimes \ell}$. Then, for any $\delta \in (0,1)$, 
%\begin{align}
%&\frac{1}{n} |\{i: H(U_i | U^{i-1}) > 1-\delta\}|  \stackrel{n \to \infty}{\longrightarrow}  H(p),
%\end{align}
%where $H(p)$ is the entropy of the distribution $p$. 
Let $X^n=[X_1,\dots,X_n]$ be i.i.d. Bernoulli($p$), $n$ be a power of 2, and $Y^n=X^n G_n$, where $G_n=    \bigl[\begin{smallmatrix} % or pmatrix or bmatrix or Bmatrix or ...
      1 & 0 \\
      1 & 1 \\
   \end{smallmatrix}\bigr]^{\otimes \log_2(n)}$ and where the addition is modulo 2. Then, for any $\e \in (0,1)$, 
\begin{align}
&|\{j \in \{1,\dots,n\}: H(Y_j | Y^{j-1}) \in (\e,1-\e) \}|  =o(n). \label{polar}
\end{align}
%where $H(p)$ is the entropy of a Bernoulli($p$) distribution. 
\end{thm}
Here $H(Y_j | Y^{j-1})$ denotes the conditional Shannon entropy of $Y_j$ given the previous components $Y^{j-1}=[Y_1,\dots,Y_{j-1}]$. 
Note that \eqref{polar} implies that, besides for a vanishing fraction, all conditional entropies $H(Y_j | Y^{j-1})$ tend to either 0 and 1. This explains the name of `polarization phenomenon'. Moreover, since $G_n$ is invertible, the total entropy in $X^n$ and $Y^n$ is the same, 
\begin{align}
nH(p)= H(X^n)=H(Y^n),
\end{align}
and
\begin{align}
&\frac{1}{n} |\{j \in \{1,\dots,n\}: H(Y_j | Y^{j-1}) \geq 1-\e \}|  \stackrel{n \to \infty}{\longrightarrow}  H(p), \label{polar1}\\
&\frac{1}{n} |\{j \in \{1,\dots,n\}: H(Y_j | Y^{j-1}) \leq \e \}|  \stackrel{n \to \infty}{\longrightarrow}  1-H(p), \label{polar2}
\end{align}
where $H(p)$ is the entropy of the Bernoulli($p$) distribution. 

%\begin{align}
%R_\e(p)=\{ j \in [n] : H(Y_j |Y^{j-1}) \geq 1 -\e \} \label{defr}
%\end{align}
%\begin{align}
%D_\e(p)=\{ j \in [n] : H(Y_j |Y^{j-1}) \leq \e \},
%\end{align}

In view of previous result, the transform $G_n$ has a clear advantage for compressing the source $X^n$: by mapping $X^n$ to $Y^n$ the compression is done by storing only the components with high conditional entropy (or more precisely with non-low conditional entropy) and the resulting compression rate is optimal since it is the entropy. 
The probability of decoding wrongly a bit which has not been stored, using the stored bits, can be seen to be at most its conditional entropy, i.e., at most $\e$. See Section \ref{rand-comp} for details about this.  
Since it shown in \cite{ari3} that $\e$ can be taken exponentially small in the square root of the blocklength while preserving the polarization theorem above, error propagation can be prevented with a successive decoding algorithm, and the above gives a linear source code with a structured matrix $G_n$ for which the encoding and decoding complexity in $O(n \log n)$ \cite{ari}. 

In addition to the above, the polarization technique has proved to admit generalizations to several multi-user problems.
In \cite{2mac,mmac} a generalization is proposed for multiple access channels.
The polar coding technique is also used in \cite{relay} for relay channels and in \cite{korada,bc} for broadcast channels. It is used in \cite{hof,gamal,rathi,vardy} for the wire-tap channel, ensuring weak secrecy, besides for noiseless main channels where strong secrecy can be achieved \cite{vardy}. In \cite{korada,ari3,corr} the Slepian-Wolf coding problem is considered with polar codes. The approach of \cite{korada,ari3} is based on the onion-pealing decoding of the users, while in \cite{corr} joint decoding of the multiple users' inputs. The joint decoding approach was also used in \cite{2mac,mmac} for the MAC. %In a preliminary version of \cite{ari3} for the two user case. 
In subsequent works, generalizations are obtained for non binary alphabets \cite{sasoglu} and for non corner points of the rate region in  \cite{arinew}. 

This paper develops the results presented first by the author in \cite{corr} and provides polar coding schemes for the $m$-user Slepian-Wolf coding problem. In this setting, multiple users have a source sequence and the goal is to compress the multiple sources in a distributed manner while exploiting the correlation among the sources (details about this problem are provided in Section \ref{compression}). 
The key step to establish polar coding schemes for the Slepian-Wolf problem is to extend Theorem \ref{thmari} to a setting where there is not a single source sequence $X^n$, but multiple sources sequences $X^n[i]$ for a set of users indexed by $i \in \{1,\dots,m\}$. There are different strategies for that purpose: one may treat each user sequentially, this is the onion-pealing approach which leads to corner points of the achievable rate region, or jointly, to reach rates on the dominant face of the achievable rate region.  The onion-pealing approach is developed in \cite{korada,ari3}, this paper provides the joint approach. 
The paper then expands the technique to the secret key generation problem in the multi-terminal setting with no side-information, which has not been considered with neither the onion-pealing nor joint approach so far\footnote{Preliminary results were presented in \cite{allerton}}. It is shown that the resulting coding schemes retain the desirable properties of polar codes. For the Slepian-Wolf (SW) coding problem, capacity-boundary rates are achieved with low-complexity complexity if the number of users is fixed. 
While in \cite{corr} for the SW problem and in \cite{2mac,mmac} for the MAC problem, the joint approach typically leads to a gap in the achievable rate region, a condition on the source distribution (separability) is introduced in this paper to ensure achievability of individual rates without gap. For the secret key generation problem, two protocols are presented affording weak secrecy, which can be converted to strong secrecy by relying on a random seed of size $o(n)$, at the optimal secrecy rate and with $O(n \log (n))$ complexity. %The achievability of secrecy is shown in a more direct way than for the wire-tap setting. 
The technique is also shown to extend to secret key agreement with more than two parties. In each of these applications, the proposed coding schemes are among the first to reach the information theoretic limits with $O(n \log (n))$ encoding and decoding complexity.

We introduce next some notation:
\begin{itemize}
\item $[n]=\{1,2,\dots,n\}$
\item $\F_2$ denotes the binary field ($GF(2)$)
\item For $x \in \F_2^k$ and $S \subseteq [k]$, $x[S]=\{x_i : i \in S\}$
\item For $x \in \F_2^k$, $x^{i}=[x_1,\dots,x_{i}]$
%\item $\{0,1,\dots,m\} \pm \e = [-\e,\e] \cup [1-\e,1+\e] \cup \dots \cup [m-\e,m+\e]$
\item $\mZ \pm \e$ denotes the set of real numbers within distance $\e$ of an integer, and $\mZ^3 \pm \e$ the set of 3-dimensional real vectors within distance $\e$ of a 3-dimensional integral vector  
%\item $H(X|Y)= \sum_y (\sum_x p_{X|Y}(x|y) \log 1/p_{X|Y}(x|y)) p_Y(y)$
\item For a probability distribution $\mu$ on $\F_2^m$ and $S \subseteq [m]$, $\mu_S$ denotes the marginal of $\mu$ on $S$.
%\item For a matrix $A$, the matrix $A^{\otimes k}$ is obtained by taking $k$ Kronecker products of $A$ with itself.
\end{itemize}

\subsection{Comment on entropy extraction}\label{rand-comp}
We discuss here the implication of extracting the Shannon entropy in a random source, with respect to reconstructing the source. In particular, this explains why the polarization phenomenon in Theorem \ref{thmari} implies a lossless source code. 
%Note the following basic lemma. For two discrete random variables $X,Y$, we denote by $P_e(X|Y)$ the average probability of error of MAP decoding, i.e. the expected probability (over $Y=y$) that the maximal probability for $X$ given $Y=y$ is not the true realization of $X$. 
%\begin{lemma}
%Let $X$ be a random variable and $f$ a map such that $H(f(X)) \geq H(X)(1-\e)$. Then $P_e(X|f(X)) \leq \e$. 
%\end{lemma}

We first define the min-entropy. 
\begin{definition}
Let $X,Y$ be random variables taking values in the discrete sets $\X,\Y$ respectively. The conditional min-entropy of $X$ given $Y$ is given by 
\begin{align}
H_\infty(X|Y) = \sum_{y \in Y} H(X|Y=y) \Pr\{Y=y\},
\end{align}
where 
\begin{align}
H_\infty(X|Y=y) = \min_{x \in \X} \log \frac{1}{\Pr\{X=x|Y=y\}}.
\end{align}
\end{definition}
Note that $H_\infty(X|Y=y) \leq H(X|Y=y)$, hence $H_\infty(X|Y) \leq H(X|Y)$, where $H$ denotes the Shannon-entropy. 

Define next a $(\mu,\e)$-compressor as follows. 
\begin{definition}
Let $\mu$ be a probability distribution on $\F_2^n$ and let $X\sim \mu$. A deterministic map $f: \F_2^n \to \F_2^m$ is a $(\mu,\e)$-compressor if $H_\infty(X|f(X)) \leq \e$. A linear compressor is a compressor for which $f$ is linear. 
%A $(p,n,\e)$-compressor is a $(\mu,\e)$-compressor where $\mu= \text{Ber}(p)^n$ is the $n$-fold product Bernoulli distribution.
\end{definition}
%A compressor is related to a condensor, in the sense that it reduces the dimension of the source while preserving the min-entropy.  
%Note that the min-entropy does not satisfy in general the chain rule $H(X|f(X)) = H(X)- H(f(X))$ obtained for the Shannon entropy. 
A compressor preserves all the information about $X$ in $f(X)$ in the following sense.  
\begin{lemma}\label{comp-lemma}
$f: \F_2^n \to \F_2^m$ is a $(\mu,\e)$-compressor then there exists an algorithm recovering $X\sim \mu$ from $f(X)$ with average error probability at most $1-e^{-\e}=\e + o(\e)$. Conversely, if for $f: \F_2^n \to \F_2^m$ there exists an algorithm recovering $X\sim \mu$ from $f(X)$ with average error probability at most $\e$, then $f$ is a $(\mu,\delta)$-compressor with $\delta=\sqrt{\e m}+ o(\sqrt{\e m})$. 
\end{lemma}
The proof of this lemma is given in Section \ref{proofA}. 

Consider now the map $g: \F_2^n \to \F_2^m$ given by $G_n[D_{\e,n}(p)^c,\cdot]$ (i.e., keep the rows in $G_n$ indexed by $[n] \setminus D_{\e,n}(p)$), where 
\begin{align}
D_{\e,n}(p) := \{i \in [n] : H(Y_i|Y^{i-1}) \leq \e \},
\end{align}
and where $Y^n=G_n X^n$, with $X^n$ i.i.d.\ Bernoulli$(p)$ as in previous section. Then Theorem \ref{thmari} says that $g$ is a $\e n$-compressor for an i.i.d.\ Bernoulli source, indeed,
\begin{align}
H(X^n|G_n[D_{\e,n}(p)^c,\cdot] X^n) &= H(Y^n| Y^n[D_{\e,n}(p)^c])\\
&=H(Y^n[D_{\e,n}(p)] | Y^n[D_{\e,n}(p)^c]) \\
& \leq \sum_{i \in D_{\e,n}(p)} H(Y_i | Y^{i-1}) \\
& \leq \e n.
\end{align}
Hence, if $\e=o(1/n)$ (and in fact $\e$ can be taken exponentially small for polar codes), we have that $G_n[D_{\e,n}(p)^c,\cdot]$ is a $o(1)$-compressor and by Lemma \ref{comp-lemma}, $X^n$ can be recovered with probability $1-o(1)$. In the case of polar codes,  $X^n$ can be recovered with a successive decoding algorithms that requires only $O(n \log(n))$ computations, as opposed to MAP (as used in the proof of the Lemma) which is a priori of exponential complexity. 

Note also that in the case of an i.i.d.\ Bernoulli$(p)$ source, the least dimension $m$ to obtain a $o(1)$-compressor is $m=nH(p)+o(n)$, i.e., the compressor must extract all the Shannon entropy. Note that the only source with binary components whose dimension is equal to its entropy is a uniformly distributed source, while here the compressed source $f(X)$ must have dimension matching approximately its total entropy, up to a $o(n)$. Hence $f$ is not necessarily an extractor as $f(X)$ may still contain some components that are not close to uniformly distributed. For example, in the polar construction, $G_n[D_{\e,n}(p)^c,\cdot] X^n$ contains also the ``moderate'' components that have conditional entropy $H(Y_i|Y^{i-1}) \in (\e,1-\e)$, and are hence non-uniform. However, modulo removing a vanishing fraction of components, $f$ can be made an extractor, e.g., $G_n[R_{\e,n}(p),\cdot]$ is an extractor for $R_{\e,n}(p) := \{i \in [n] : H(Y_i|Y^{i-1}) \geq 1-\e \}$, as shown in \cite{extract}.

%For a discrete random variable $X$, the probability of guessing $X$ wrongly is minimized by  choosing the most likely value of $X$, i.e., $P_e(X) =1 - \max_{x}P_X(x)$, hence $P_e(X)  \leq 1 - 2^{-H(X)} \leq H(X)$. Applying this inequality to a conditional distribution $P_{X|Y=y}$ yields the result for the average probability of error of guessing $X$ when observing $Y$ using MAP decoding: $P_e(X|Y) \leq H(X|Y)$.  

\subsection{Comment on the use of the bit-reversal matrix}
The transformation used \cite{ari,ari3} is not $G_n$ but $ G_n B_n$, where $B_n$ is the bit-reversal permutation (i.e., the permutation that maps each element of $j \in \{1,\dots,n\}$ to the element $\tilde{j}$ obtained by reversing the binary expansion\footnote{One should consider here $j$ in $\{0,\dots,n-1\}$ rather that $\{1,\dots,n\}$.} of $j$.)
We clarify here why $B_n$ is not necessary for obtaining Theorem \ref{thmari}, but is simply {\it convenient} to prove the result. 
In fact, the probability distributions of $Y^n= X^n G_n$ and $\tilde{Y}^n= X^n G_n B_n$ are the same when $X^n$ is i.i.d.\ Bernoulli$(p)$, in particular 
%Consider $X^n$ to be i.i.d.\ Bernoulli$(p)$, and define $Y^n= X^n G_n$ and $\tilde{Y}^n= X^n G_n B_n$. It turns out that 
\begin{align}
H(Y_j | Y^{j-1}) = H(\tilde{Y}_j | \tilde{Y}^{j-1}) , \quad \forall j \in \{1,\dots,n\}. \label{idi}
\end{align}
%Hence the components of $\tilde{Y}^n$ polarize in the exact same way as the components of $Y^n$. 
To see this, observe that $G_n$ and $B_n$ commute. This is easily seen by expressing the components of $G_n$ as Boolean functions of the row and column index (in binary expansions), namely, for $a,b \in \{0,1\}^{\log_2 n}$, where $a$ represent the row index and $b$ the column index, $G_n(a,b)= \prod_{i=1}^n (1+(a_i  + 1)b_i)$, where the addition is modulo 2. Therefore, reversing the order of the $\{a_i\}$'s or the order of the $\{b_i\}$'s leads to the same matrix. Since $B_n$ and $G_n$ commute, and since $B_nX^n$ has the same distribution as $X^n$, namely i.i.d.\ Bernoulli$(p)$, the equality \eqref{idi} follows.

The reason why $B_n$ might be preferred in polar coding is that the ``recursive'' nature of the polarization process may be more easily described with $B_n$. To see this, consider $n=4$, and
\begin{align}
G_4 = \begin{pmatrix} 1 &0 & 0 & 0 \\ 1 &1 & 0 & 0 \\ 1 &0 & 1 & 0 \\ 1 &1 & 1 & 1 \end{pmatrix},
\end{align}
$Y^4=X^4G_4$ and $\tilde{Y}^4=(Y_1,Y_3,Y_2,Y_4)$.
Observe that $\tilde{Y}^4$ can be constructed recursively as follows:
Define 
\begin{align}
&V_1=X_1+X_2, \quad 
V_2=X_2\\
&W_1=X_3+X_4, \quad
W_2=X_4.
\end{align} 
Then 
\begin{align}
&\tilde{Y}_1=V_1+W_1,\quad
\tilde{Y}_2=W_1\\
&\tilde{Y}_3=V_2+W_2,\quad
\tilde{Y}_4=W_2. 
\end{align} 
If one does not use the bit-reversal, the following is obtained, 
\begin{align}
&Y_1=V_1+W_1, \quad
Y_2=V_2+W_2, \\ 
&Y_3=W_1, \quad Y_4=W_2. 
\end{align} 
Hence the recursion is seen differently, in $\tilde{Y}^4$, two consecutive components are connected by $G_2$, whereas in $Y^4$, two alternating components are connected by $G_2$. 
Yet, as shown by \eqref{idi}, the two orderings lead to the same conditional entropies, this is simply seen here by noting that $Y_2=X_2+X_4$ and $Y_3=X_3+X_4$ can be swapped up to a relabelling of the components. To perform the decoding algorithm recursively in $O(n \log n)$, it is usually more convenient to use the bit-reversal at the encoder or decoder. 
%Of course, this is purely a notational convenience\footnote{The bit-reversal matrix makes also the description of the decoding algorithm easier.}, since \eqref{idi} holds.

\section{Results}
A joint polarization phenomenon for correlated sources is now presented. Implications of this result for the Slepian-Wolf coding and secret key generation problems are next presented. 

\subsection{Matrix polarization}

\begin{thm}\label{main}
Let $m$ be a positive integer, $n$ be a power of 2 and $X^n$ be an $m \times n$ random matrix with i.i.d.\ columns of distribution $\mu$ on $\F_2^m$.
Let $Y^n=X^n G_n$ over $\F_2$, where $G_n=\bigl[\begin{smallmatrix} % or pmatrix or bmatrix or Bmatrix or ...
      1 & 0 \\
      1 & 1 \\
   \end{smallmatrix}\bigr]^{\otimes \log_2(n)}$. 
\begin{itemize}
\item For any $\e>0$, $$|\{i \in [n] :  H(Y_i[S]|Y^{i-1}) \notin \mZ \pm \e, \text{ for any } S \subseteq [m] \}| = o(n).$$
\item For any $\e>0$, for any $i \in [n]$, there exists $A_i \in \F_2^{m \times m}$ such that 
\begin{align*}
& H(A_i Y_i |Y^{i-1}) \leq \e,
& \frac{1}{n} \sum_{j=1}^n \nullity(A_j) \stackrel{n \to \infty}{\to} H(\mu).
\end{align*}
\item Previous statement hold when $\e=O(2^{-n^\beta})$, $\beta < 1/2$. 
\end{itemize}
\end{thm}
The above is a counter-part in the source setting of the results established in \cite{mmac} for MACs. The proof is slightly more direct in the source setting, since there is no freezing of the noisy components to be done. Moreover, the scheme is entirely deterministic in the source setting (there is no randomisation of the frozen components). 

The first item in the theorem says that for all $i$ but a vanishing fraction, the conditional entropies $H(Y_i[S]|Y^{i-1})$ take near-integer values for all $S$. In particular, one would like to conclude that there is a set $S$ of maximal cardinality for which $H(Y_i[S]|Y^{i-1})=H(Y_i|Y^{i-1})$ is close to $|S|$, which implies that the components $Y_i[S]$ are roughly i.i.d.\ Bernoulli$(1/2)$ and that $Y_i[S^c]$ is roughly deterministic given $Y^{i-1}$ and $Y_i[S]$. This is implied by the second item, which gives even more information on the structure of $Y_i$.  
The second item says that there is a transformation (the matrix $A_i$) which turns the components of $Y_i$ into a lower-dimensional almost deterministic random vector (given the past components). Of course it is always possible to find such a transformation $A_i$, e.g., pick the 0 matrix. However, the second condition in item 2 implies that this is achieved without loosing any information, i.e., preserving all the entropy of the original random matrix. This linear structure also shows that there are various ways to chose the components of $Y_i$ that are random or deterministic, and the alternatives are governed by the structure of $A_i$, since there may be various ways to select a set of independent columns in $A_i$.

\subsection{Compression results}\label{cr}
Theorem \ref{main} can be used in various compression settings, in particular to compress sources which have finite field alphabets or which have finite memory, as shown in \cite{corr}. 

\begin{remark}It easy to see how this can be achieved. If the source has finite memory, group the components which are independent to create a matrix $X^n$ which has independent columns as in Theorem \ref{main}. Compression is then obtained by retaining only the high entropy components in $Y^n$. 
For the case of source with alphabet cardinality $2^m$, apply the matrix $G_n$ to the source with the $GF(2^m)$ field addition. The theorem then leads to a generalized notion of polarization where the $2^m$-bit components $Y_i$ are not uniform or deterministic, but have subsets of components (looking at $Y_i$ as an elements of $GF(2)^m$) which are uniform or deterministic. More details about these approaches can be found in \cite{corr,mmac}. %We provide discuss the Slepian-Wolf coding problem.
\end{remark}

%Note that, using \eqref{max} for the definition of $S_j$ (and the corresponding $D_\e$), the realizations of $Y^{j-1}$ and $Y_j[S_j^c]$ are known, and with high probability one guesses $Y_j[S_j]$ correctly in step 1, because of \eqref{corr}.  
%Moreover, due to the Kronecker structure of $G_n$, and similarly to \cite{ari}, step 1.\ and the entire algorithm require only $O(n \log n)$ computations. 
%Finally, from the proof of Theorem \ref{main} part (2), it results that step 1. can also be performed slightly differently, by finding sequentially the inputs $Y[j]$ for $j \in S_j$, reducing an optimization over all possible $y \in \F_2^{|S_j|}$, where $|S_j|$ can be as large as $m$, to only $m$ optimizations over $\F_2$ (which may be useful for large $m$). 

We provide now results for the Slepian-Wolf problem using the coding scheme introduced in Section \ref{compression}. Further details about the Slepian-Wolf coding problem are also provided in Section \ref{compression}.  The following definition is needed.

\begin{definition}
A probability distribution $\mu$ on $\F_2^m$ is {\it separable in $S$} if  
%\begin{align}
%\mu_S = \sum \mu \star \zeta_i
%\end{align}
%In other words, 
for $(X[1], \dots, X[m]) \sim \mu$, there exist a linear map $F: \F_2^{m-|S|} \to \F_2^{|S|} $ and a random variable $W[S]$ independent of $X[S^c]$ such that 
\begin{align}
X[S] = F(X[S^c])+W[S].
\end{align}
\end{definition}

For example, if $m=2$, $X_1$ is uniform and $X_2=X_1+Z$ where $Z$ is independent of $X_1$ but otherwise arbitrary, then $\mu$ is separable in both components. 
%If $m=3$, $X_1$ is uniform, $X_2=X_1 + Z_1$ and $X_3=X_2+Z_2$, where $Z_1,Z_2$ are independent of $X_1,X_2$, then $\mu$ is separable in every subset of components: by definition, $X_2=X_1 + Z_1$, so taking $F$ to be the projection of $(X_1,X_3)$ onto the first component, we obtain the separability condition. Moreover,  $X_1=X_2+Z_1$, and since $Z_1$ is independent of $X_1$ and $X_1$ is uniform, $Z_1$ is also independent of $X_2$, and $(X_1,X_2)= (X_3,X_3) + (Z_1+Z_2,Z_2)$, etc.\ Any subset is hence separable. In general, one may only have certain subsets which are separable. 
Another case is when $F$ is identically 0, then $X[S]$ and $X[S^c]$ are independent, a strong case of separability. 

\begin{thm}\label{sw}
Let $m \geq 1$ be fixed. Let $n$ a power of 2, and let $X^n[1],\dots,X^n[m]$ be $m$ sources with $(X_i[1],\dots,X_i[m])_{i=1}^n$ i.i.d.\ under $\mu$ on $\F_2^m$. The coding scheme of Section \ref{compression} allows to achieve distributed compression of the $m$ sources at total sum-rate $H(\mu)$, error probability $O(2^{-n^{\beta}})$, $\beta < 1/2$, and encoding/decoding complexity of $O(n \log n)$. 
Further, if $\mu$ is separable in $S \subseteq [m]$, the encoders in $S$ achieve the sum-rate on $S$, namely $H(X_1[S]|X_1[S^c])$.  
%For $m$ correlated sources of joint distribution $\mu$, previously described scheme allows to perform lossless and distributed compression of the sources at sum-rate $H(\mu)$, with an error probability of $O(2^{-n^{\sqrt{\beta}}})$, $\beta < 1/2$, and an encoding and decoding complexity of $O(n \log n)$.
\end{thm}
%Note that this result allows to achieve the sum-rate of the Slepian-Wolf region, i.e., a rate belonging to the dominant face of the Slepian-Wolf achievable rate region, it does not say that any rate in that region can be reached with the proposed scheme.

%Note also that for any distribution $\mu$ of $X=(X[1], \dots, X[m])^T$, if one considers a linear transformation of this vector $X'=MX$, the polarization pattern of $X'$ can be obtained directly from the polarization pattern of $X$ since $G_n$ is multiplied to the right of $X'$. A particular case of interest is when $\mu$ is the distribution of $MZ$ where $Z$ are $m$ independent (but not necessarily i.d.) binary random variables and $M$ is an $m \times m$ binary matrix. Then one can apply the single user polarization for each row and the entire rate region is achieved. This means that if a distribution is separable (like for $Z$), a linear transformation on the left does not cause loss of rate.  

Theorem \ref{sw} always ensures achievability of the sum-rate of all users, however the sum-rate on subsets of users is in general not achieved.  
%If the distribution $\mu$ is the distribution of $MZ^m$ where $Z^m$ are $m$ independent (but not necessarily i.d.) binary random variables and $M$ is an $m \times m$ binary matrix, then the entire rate region is achieved. This is indeed a distribution separable in all components\footnote{An alternative way to see that the full rate region is preserved for such distribution is to express $X$ as $MZ$ where $M$ is an $m \times m$ binary matrix and $Z$ has independent rows which are i.i.d.. Then one can use the single user polarization for each row.}. 
How much of the rate region is lost in the general case remains an open problem. 

%Note that this result allows to achieve the sum-rate of the Slepian-Wolf region, i.e., a rate belonging to the dominant face of the Slepian-Wolf achievable rate region, it does not say that any rate in that region can be reached with the proposed scheme. 
 
\subsection{Secrecy results}
We provide in this section the performance achieved with the two protocols presented in Section \ref{secrecy}. The first protocol uses the onion-pealing polarization \cite{ari3} whereas the second protocol uses Theorem \ref{sw}. Further details about the secret key generation setting are also provided in Section \ref{secrecy}. 

\begin{thm}\label{sk}
Let Alice observe $X^n$ and Bob observe $Y^n$ such that $(X^n,Y^n) \iid \mu$, where $\mu$ is a probability distribution over $\F_2 \times \F_2$. The polar-key-1 and polar-key-2 protocols described in Section \ref{secrecy} allow Alice and Bob to generate respectively private keys $S_n$ and $S_n'$ using public communication\footnote{Public communication is assumed to take place over a noiseless channel.} $C_n$ such that for any $\beta<1/2$ and $n$ large enough
\begin{align}
&\pp\{ S_n \neq S_n'\} = o(2^{-n^\beta}),\\ 
&I(S_n; C_n) =o(n), \\
&H(S_n) 
 \begin{cases}  = n I(\mu) + o(n), & \text{for polar-key-1},\\
   \geq n \frac{1}{2} I(\mu) + o(n), & \text{for polar-key-2},
\end{cases}
\end{align}
where $I(\mu)=I(X_1;Y_1)$. The protocols require one communication round and their computational complexity is $O(n\log(n))$.
\end{thm}
\begin{remark}
The theorem says that Alice and Bob can w.h.p.\ agree on a secret key of rate $I(\mu)$ which leaks $o(n)$ bits of information to an eavesdropper who listens to the public communication. This provides weak secrecy. 
%However, a direct approach to achieve strong secrecy solely based on the polarization phenomena seems difficult. In words, 
The difficulty in achieving strong secrecy lies in the fact that for secret key agreement, the reliability of the key construction relies on transmitting the non-low entropy components whereas the secrecy relies on hiding the non-high entropy components. However there are $o(n)$ fluctuating components in the polarization phenomenon which are in neither camp and hence compromise either reliability (block error probability) or strong secrecy. 
On the other hand, if these $o(n)$ components are securely transmitted, e.g., using a one-time pad, then strong secrecy is achieved using the previous protocols.
Another alternative to achieving strong secrecy is via privacy amplification \cite{amplification}. 
\end{remark}
 
The two protocols afford several desirable attributes: 
\begin{enumerate}
\item the computational complexity is $O(n \log n)$ 
\item the key construction is deterministic  
\item the information-theoretic limit is achieved (for any distribution $\mu$ with the first protocol and for some\footnote{This is the case for example when $X$ and $Y$ are connected by a BSC} distributions $\mu$ with the second protocol). 
\end{enumerate}

%In \cite{hof,gamal,rathi,vardy} similar constructions to the first protocol are obtained for the wire-tap channel. 
%In \cite{vardy} strong secrecy is obtained for the special case where $X=Y$ (which is of interest only when the eavesdropper has side-information). 
%In this paper, we recover the same performance with the first protocol. 
The second protocol provides a different approach for secret key agreement, as it performs the decoding differently. The variation is of interest for the error-probability performance, since bits are decoded jointly rather than with the onion-pealing approach, we expect the finite block length performance to be different.  We also provide extensions to secret key agreement with more than two parties. 

\section{Compression}\label{compression}

%\subsection{Slepian-Wolf coding}
%A direct approach is left open for future work; we investigated this here, for arbitrary many users.   

The Slepian-Wolf coding problem consists in compressing correlated sources without the encoders cooperating after the code agreement.  
%Namely, consider $m$ binary sources which are correlated with an arbitrary distribution $\mu$. We are interested in compressing an i.i.d.\ output of these sources. 
Let $X_1,\dots,X_n$ be i.i.d.\ under $\mu$ on $\F_2^m$, i.e., $X_i$ is an $m$ dimensional binary random vector and $X_1[i],\dots,X_n[i]$ is the sources output for user $i$.
Encoding these sources by having access to all the realizations requires a rate $H(\mu)$ (and it is the lowest achievable rate).
In \cite{slepian}, Slepian and Wolf showed that, even if the encoders are not able to cooperate while observing the source realizations, lossless compression can still be achieved at sum rate $H(\mu)$.

In \cite{ari3, korada}, polar codes are used for the two-user Slepian-Wolf coding problem by reducing the problem
to single-user source coding problems, using ``onion-pealing decoding'' and hence achieving a corner point of the rate region. 
The method is based on the following extension of \eqref{polar},
\begin{align}
&\frac{1}{n} |\{j \in [n]: H(Y_j[1] | Y^{j-1}[1], Y^n[2]) \geq 1-\e \}| \\
& \phantom{\frac{1}{n} |\{j \in [n]: } \stackrel{n \to \infty}{\longrightarrow}  H(X[1]|X[2]), \label{polar2}
\end{align} 
where $(X[1],X[2]) \sim \mu$.  
This means that conditioning on the {\it entire} random vector $Y^n[2]$, the first random vector polarizes.  The joint approach is also partly discussed in a preliminary version of \cite{ari3} for two correlated sources. We now present how to use Theorem \ref{main} to achieve rates on the dominant-face jointly for arbitrarily many users. 

%\begin{definition}
%A random variable $Z$ over $\F_2^k$ is $\e$-uniform if $H(Z) \geq k(1- \e)$, and it is $\e$-deterministic if $H(Z) \leq \e k$. We also say that $Z$ is $\e$-deterministic given $W$ if $H(Z|W) \leq \e k$.
%\end{definition}

From Theorem \ref{main}, we can associate to each component $j\in [n]$ a matrix $A_j$ such that $$H(A_j Y_j |Y^{j-1}) \leq \e.$$ If $A_j$ has rank $r_j$, by revealing $k_j=m-r_j$ components in $Y_j$ appropriately, denoted by $B_j \subseteq [m]$, we have that $A_j Y_j$ can be reduced to a full rank matrix multiplication $\tilde{A}_j Y_j[B_j^c]$ plus an almost deterministic vector. There may be several choices of $B_j$ that work, as there may be several subsets of columns in the matrix $A_j$ which have full rank. Any choice works for achieving the sum-rate. Now, $Y_j[B_j^c]$ is almost deterministic given $Y^{j-1}$ and $Y_j[B_j]$. 
Hence the number of bits to reveal is exactly $\sum_j k_j$, and as stated in the lemma, this corresponds to the total entropy of $Y$ (up to a $o(n)$).

\begin{definition}\label{bk}
For $\e>0$ and $i\in [n]$, let $A_i$ be a matrix as in Theorem \ref{main}. 
Select $B_i=B_{i,\e} \subseteq [m]$ such that $|B_i|$ has maximal cardinality and such that the columns of $A_i$ indexed by $B_i$ are linearly independent. For $k\in [m]$, define $B[k]=B_\e[k]$ to be the subset of $[n]$ containing the indices $j$ for which $k \in B_j$.  
\end{definition}

\noindent
{\bf Polar codes for distributed data compression:}\\
1. For a given $n$ and $\e$ (which sets the error probability), 
the users agree on a ``chart'' of deterministic indices, i.e., $B[k]$ for user $k$ as in Definition \ref{bk}.\\ %Each user identifies its own chart $D_\e(i,\cdot)$. \\
2. For $k \in [m]$, user $k$ computes $Y[k]=X[k] G_n$ and transmit $Y_{B_{\e}[k]^c}[k]$ to the decoder. \\
3. The decoder, in possession of $Y_{B_{\e}[k]^c}[k]$ for every $k \in [m]$, runs {\it polar-matrix-decoder} to get $Y$. 
% with an error probability at most $\e n$. 
Since $G_n$ is invertible, $G_n^{-1}=G_n$, the decoder finds $X=Y G_n$.

\begin{definition}{\it polar-matrix-decoder}\\
Inputs: for $k \in [m]$, $B[k] \subseteq [n]$ and $y_{B[k]}[k] \in \F_2^{|B[k]|}$.\\
Output: $y \in \F_2^{m \times n}$.\\
Algorithm:\\
For $j=1,\dots,n$, let $B_j^c=\{k : j \notin B[k]\}$, if $B_j^c$ is not empty, compute 
$$\hat{y}[B_j^c]=\arg\max_{u \in \F_2^{|B_j^c|}} \pp\{ Y[B_j^c]= u | Y^{j-1}=y^{j-1}, Y[B_j]=y[B_j] \}.$$ 
%0. Let $M=B$;\\
%1. Find the smallest $j$ such that $S_j=\{(i,j) \in M\}$ is not empty; compute $$\hat{y}[S_j]=\arg\max_{u \in \F_2^{|S_j|}} \pp\{ Y[S_j]= u | Y^{j-1}=y^{j-1}, Y[S_j^c]=y[S_j^c] \};$$
%2. Update $M=M \setminus \{j\}$, $y[M]=y[M] \cup \hat{y}[S_j]$; \\
%3. If $M$ is empty output $y$, otherwise go back to 1. 
\end{definition}

\begin{remark} The results in this paper may hold when $m=o(n)$, but one has to be careful with the complexity scaling when $m$ gets large. An interesting question is to investigate whether the decoding complexity can be brought down from $2^m \log m$ to $m$ with the joint decoding. This is possible if the source is separable in all components, as in the example of Section \ref{cr}. 
%In general it is not clear if this can be achieved.  
\end{remark}

It would also be interesting to study the gap in the marginal rates of the capacity region achieved with the joint approach. In a subsequent work, \cite{arinew} proposes an alternative approach to reach arbitrary rates on the dominant faces. The approach does not rely on a polarization of the rate region but on variations of the entropy chain-rules. It remains an open problem to find a method to polarize the entire rate region into extremal rate regions without loosing individual rates. We next show that such a decomposition, irrespectively of the polar coding technique, is possible when $m\leq 3$, but may not be achievable for any distribution when $m\geq 4$.

\begin{definition}
For a probability distribution $\mu$ on a finite set, define $\rho_\mu$ to be the $2^m$-dimensional real vector formed by the collection of numbers $\rho_\mu(S)=H(X[S]|X[S^c])$, where $S \subseteq [m]$ and $X[1,\dots,m] \sim \mu$. We call such vectors {\it Slepian-Wolf vectors}.    
\end{definition}

\begin{prop}\label{ray}
For $m \leq 3$, $\rho_{\mu}$ can always be expressed as a positive weighted sum of extremal Slepian-Wolf vectors generated by a linear form as in Theorem \ref{main}. For $m \geq 4$, there exist distributions $\mu$ for which the decomposition is not possible.    
\end{prop}

\begin{proof}[Proof of Proposition \ref{ray}]
Note that $\rho_\mu(S)=H(X[S]|X[S^c])=H(\mu)-H(X[S^c])$. For a fixed $\mu$, consider the {\it entropic vectors} whose components are given by the entropies $H(X[S])$ for all $S \subseteq [m]$. 
Define a polymatroid vector $f$ on a ground set $[m]$ to be a positive vector of length $2^m$ where the components $f(S)$, indexed by subsets $S \subseteq [m]$, satisfy
\begin{align}
& f(\emptyset)=0,\\
&f(J) \leq f(K), \quad  \forall J \subseteq K \subseteq [m], \\
&f(J \cup K)  + f(J \cap K) \leq f(J) + f(K) , \quad \forall J,K \subseteq [m].
\end{align}
Note that since $f(\emptyset)=0$ by definition, we will only specify the $2^m-1$ components of a polymatroid/entropic vector which correspond to the non-empty sets in what follows.  
Denoting by $\Gamma_m^*$ the set of all possible entropic vectors on $m$ random variables, and by $\Gamma_m$ the set of all polymatroid vectors on a ground set $[m]$, it is know from Theorem 2 page 1984 in \cite{yeung} that $\bar{\Gamma}_3^*=\Gamma_3$. 
Moreover, $\Gamma_3$ is a convex cone whose extreme rays are given by 
\begin{align*}
m_1=(1,0,0,1,1,0,1), \qquad m_2=(0,1,0,1,0,1,1),\\
m_3=(0,0,1,0,1,1,1),\qquad m_4=(1,1,0,1,1,1,1),\\
m_5=(1,0,1,1,1,1,1),\qquad m_6=(0,1,1,1,1,1,1),\\
m_7=(1,1,1,1,1,1,1),\qquad m_8=(1,1,1,2,2,2,2),
\end{align*}
where the above vectors correspond to $$(H(X_1),H(X_2), H(X_3), H(X_1,X_2), H(X_1,X_3), H(X_2,X_3), H(X_1,X_2,X_3)).$$ 
These can be checked to be the only extreme rays, using the characterization of extreme rays in Section 10.4.B of  \cite{welsh} or from \cite{matus-ray}. 
Note that $(1,1,1,2,2,2,3)$ is not an extreme ray since it corresponds to $m_1+m_2+m_3$. 
Now, any entropic vector can be expressed as a positive weighted sum of extreme rays, which in this case can be checked to be all equivalent to the rank vector of a matrix (where a rank vector is defined by $r(S)=\rank(A[S])$, $S \subseteq [m]$). For example, $(1,0,0,1,1,0,1)$ is the rank vector of  $\begin{pmatrix} 1& 0& 0  \end{pmatrix}$ and $(1,1,1,2,2,2,2)$ is the rank vector of $\begin{pmatrix} 1& 0 & 1 \\ 0 & 1 & 1 \end{pmatrix}$.

For $m\geq 4$, $\bar{\Gamma}_m^* \subsetneq \Gamma_m$. In particular, for $m=4$, there exists an extreme ray $$(1,1,1,2,2,2,2,2,2,2,2,2,2,2,2)$$ which is an entropic vector but which is not a rank function.  
It is an entropic vector by defining $Z_1,Z_2,Z_3$ i.i.d.\ Bernoulli$(1/2)$ and $X_1:=Z_1$, $X_2:=Z_2$, $X_3:=Z_1+Z_2$ and $X_4:=(Z_1,Z_2)$. It is not a rank function since $r(4)>1$.  
\end{proof}
\section{Secrecy}\label{secrecy}

We now consider the problem of information-theoretic secret key agreement. We focus first on a basic setting: Alice observes privately $X^n$ and Bob observes privately $Y^n$, where $(X_i,Y_i)$ are i.i.d.\ from a distribution $\mu$ over $\F_2 \times \F_2$. Using public communication, Alice and Bob want to agree on a secret key which has largest possible rate. Denoting by $C$ the communication between Alice and Bob and $S$ the secret key, we would like $I(S;C)$ to be small. If $I(S;C) \to 0$, strong secrecy is achieved, if $\frac{1}{n}I(S;C) \to 0$, weak secrecy is achieved. It is known from \cite{AC,maurer} that the secrecy capacity for this setting is $I(\mu)=I(X_1;Y_1)$. We now introduce the polar code protocols. 
%This paper shows that the capacity can be achieved with polar coding schemes for weak secrecy and any $\mu$, with a low-complexity construction. 

%The protocols rely strongly on \cite{source,corr,extract}. 
%The following secret key generation schemes were first presented in . 
%In \cite{hof,gamal,rathi,vardy} the wire-tap channel is considered with polar codes and results are achieved for weak secrecy. In \cite{vardy} strong secrecy is obtained for the special case where $X=Y$ (which is of interest only when the eavesdropper has side-information). In this paper, we recover the same performance with the first protocol. We then introduce a second protocol for secret key agreement based on \cite{}, which performs differently in terms of rate-achievability but which also performs differently for the error probability, as it decodes bits jointly instead of with the onion-pealing approach.  

%Although we focus here on the basic model, our approach allows extensions to multiple users and eavesdroppers with side-information using \cite{ari3,corr,mmac}. 

\subsection{First protocol}
\begin{definition}\label{sets}
For $\e>0$, $n \geq 1$ and $\mu$ a probability distribution on $\F_2 \times \F_2$, the following sets are defined 
\begin{align}
R_{\e,n}(X) &:=\{i \in [n] : H(U_i| U^{i-1})  \geq \e \},\\
R_{\e,n}(X|Y) &:=\{i \in [n] : H(U_i| U^{i-1},V^n)  \geq \e \},
\end{align}
where $U^n:= X^nG_n$, $V^n:= Y^nG_n$, $G_n=    \bigl[\begin{smallmatrix} % or pmatrix or bmatrix or Bmatrix or ...
      1 & 0 \\
      1 & 1 \\
   \end{smallmatrix}\bigr]^{\otimes \log_2(n)}$, $(X^n,Y^n) \iid \mu$.
\end{definition}

\noindent
{\bf Polar-key-1 protocol.} The protocol is defined for a given distribution $\mu$ on $\F_2 \times \F_2$ (the distribution correlating the inputs of Alice and Bob) and a given $n\geq 1, \beta \in (0,1/2)$.\\
Alice input: $x^n \in \F_2^n$, $\e \in (0,1)$, $\mu$. Alice output: $s_n$.\\
Bob input: $y^n \in \F_2^n$, $\e \in (0,1)$, $\mu$. Bob output: $s'_n$.\\
Let $\e:=2^{-n^\beta}$.\\
\noindent
(1) Alice computes $u^n=x^nG_n$ and sends $u^n[R_{\e,n}(X|Y)]$ to Bob using public communication and sets $s_n=u^n[R_{1-\e,n}(X) \setminus R_{\e,n}(X|Y)]$ (cf.\ definition \ref{sets}).\\
(2) Bob computes $\hat{u}^n$ from $u^n[R_{\e,n}(X|Y)]$ and $y^n$ using the decoding algorithm in \cite{ari3} Section III and sets $s'_n=\hat{u}^n[R_{1-\e,n}(X) \setminus R_{\e,n}(X|Y)]$.

\subsection{Second protocol}
\begin{definition}\label{sets2}
For $\e>0$, $n \geq 1$ and $\mu$ a probability distribution on $\F_2 \times \F_2$, the following sets are defined 
\begin{align}
R_{\e,n}(X) &:=\{i \in [n] : H(U_i| U^{i-1})  \geq \e \},\\
Q_{\e,n}(X|Y) &:=\{i \in [n] : H_i  \notin \mZ^3 \pm \e \} \cup \{i \in [n] : H_i  \in \{(1,0,1),(1,1,2)\} \pm \e \} 
%T_{\e,n}(X|Y) &:= \{i \in [n] : H(U_i | U^{i-1} V^{i-1} V_i) \geq \e \}.
\end{align}
%$R_{\e,n}(X|Y) :=\{i \in [n] : H_i  \notin \mZ \pm \e \} \cup \{i \in [n] : H_i  \in \{(1,0,1),(1,1,2)\} \pm \e \}$
%\begin{align*}
%%&R_{\e,n}(X,Y;001) :=\{i \in [n] : H(U_i |V_i, U^{i-1}, V^{i-1}) \leq \e, \\ & \quad H(V_i  | U_i, U^{i-1}, V^{i-1})  \leq \e , \\ &\quad  H(U_i ,V_i | U^{i-1}, V^{i-1})  \geq 1-\e\},\\
%&R^{A}_{\e,n}(X,Y) :=\{i \in [n] : H(U_i |V_i, U^{i-1}, V^{i-1}) \geq 1-\e,  \\ & \quad H(V_i  | U_i, U^{i-1}, V^{i-1})  \leq \e \},\\
%%&R_{\e,n}(X,Y;011) :=\{i \in [n] : H(U_i |V_i, U^{i-1}, V^{i-1}) \geq 1-\e, \\ & \quad H(V_i  | U_i, U^{i-1}, V^{i-1})  \leq \e \},\\
%&R^B_{\e,n}(X,Y) :=\{i \in [n] : H(U_i ,V_i | U^{i-1}, V^{i-1})  \geq 2-\e \},\\
%&T_{\e,n}(X) := R_{\e,n}(X),\\
%&T_{\e,n}(X|Y) := R^A_{\e,n}(X,Y) \cup R^B_{\e,n}(X,Y) ,
%\end{align*}
where $U^n:= X^nG_n$, $V^n:= Y^nG_n$, $G_n=    \bigl[\begin{smallmatrix} % or pmatrix or bmatrix or Bmatrix or ...
      1 & 0 \\
      1 & 1 \\
   \end{smallmatrix}\bigr]^{\otimes \log_2(n)}$, $(X^n,Y^n) \iid \mu$, and 
 \begin{align}  
   H_i=[H(U_i  | U^{i-1}, V^{i-1}, V_i) , H(V_i | U^{i-1}, V^{i-1}, U_i), H(U_i ,V_i | U^{i-1}, V^{i-1})].
   \end{align}
\end{definition}
Note that $Q_{\e,n}(X|Y)$ contains the components which are not polarized or which have high conditional entropy in the $U$ vector for Bob. 

\noindent
{\bf Polar-key-2 protocol.} The protocol is defined for a given distribution $\mu$ on $\F_2 \times \F_2$ (the distribution correlating the inputs of Alice and Bob) and a given $n\geq 1, \beta \in (0,1/2)$. Before the protocol starts, the two parties identify the marginal rates $(R_1,R_2)$ achieved with Theorem \ref{main} and find the maximum between $R_1-H(X_1|Y_1)$ and $R_2-H(Y_1|X_1)$. If the maximum of these two terms is the second term, invert the role of Alice and Bob in the following protocol. \\
Alice input: $x^n \in \F_2^n$, $\e \in (0,1)$, $\mu$. Alice output: $s_n$.\\
Bob input: $y^n \in \F_2^n$, $\e \in (0,1)$, $\mu$. Bob output: $s'_n$.\\
Let $\e:=2^{-n^\beta}$.\\
\noindent
(1) Alice computes $u^n=x^nG_n$ and sends $u^n[Q_{\e,n}(X|Y)]$ to Bob using public communication and sets $s_n=u^n[R_{1-\e,n}(X) \setminus Q_{\e,n}(X|Y)]$ (cf.\ definition \ref{sets2}).\\
(2) Bob computes $\hat{u}^n$ from $u^n[Q_{\e,n}(X|Y)]$ and $y^n$ using the decoding algorithm in \cite{corr} and sets $s'_n=\hat{u}^n[R_{1-\e,n}(X) \setminus Q_{\e,n}(X|Y)]$.

\begin{remark} The two protocols provide different secret key rates, since the rates achieved by Theorem \ref{main} are not necessarily the entire rate region of the Slepian-Wolf coding problem, whereas \cite{ari3} always achieves a corner point of the region.  Note however that, as opposed to the Slepian-Wolf coding problem where the decoder needs to decode the second-user sequence $v^n$ into an estimate $\hat{v}^n$, in the secret-key generation problem, Bob is already in possession of true sequence $v^n$. Therefore, the first protocol may not only be sub-optimal in terms of rate but it is also worse in terms of error probability, as it does not make use of the future values of $v^n$ in the decoding algorithm. Note, however, that the second protocol decodes the components of $u^n$ and $v^n$ causally without requiring the future of $v^n$. For the Slepian-Wolf coding problem the difference is  notable, since $v^n$ is not known but estimated by the decoder, and in fact the onion-pealing decoding of $u^n$ from $\hat{v}^n$ appears to be more sensitive to error propagations than the joint decoding approach at finite block length \cite{saed}. However, for secrecy, the second protocol is mainly a toy application of Theorem \ref{main}.

%Moreover, the decoding algorithm of \cite{ari3} and the {\it polar-matrix-decoder} differ in the order of which bits are decoded. In \cite{ari3}, the components of $u^n$ are first decoded successively followed by the components of $v^n$ conditioning on $u^n$. Here we decode the pairs of components $(u_i,v_i)$ successively. This slightly modifies the decoding complexity and the error probability. For two users the distinction is minor in terms of scaling, but it is expected to matter for the finite block length performance. %It is an interesting question to investigate form closer the finite block length performance of these two approaches. 
\end{remark}

\subsection{Extensions}
\subsubsection{Side-information}
One can assume that Eve receives side-information $Z^n$ about $X^n,Y^n$. In the case where $Z^n$ is i.i.d.\ and stochastically degraded form $Y^n$, such that $X^n-Y^n-Z^n$ form a Markov chain, the above protocols can be modified to generate secret keys. For example for the first protocol, require Alice to send the components of $U^n$ such that $H(U_i|U^{i-1},Y^n)\geq \e$. Eve has now $Z^n$, however by the Markov property $$H(U_i|U^{i-1},Z^n) \geq  H(U_i|U^{i-1},Z^n,Y^n)=H(U_i|U^{i-1},Y^n)$$
hence the high entropy components are nested and previous expansions apply, using for the key the components with high conditional entropy $H(U_i|U^{i-1},Z^n)$ which are missing to Eve. 
Note that for the case where Eve has degraded side-information, one may consider the special case for which $X=Y$, in this case strong secrecy is directly achieved with previous expansions.
%, as obtained in \cite{vardy} for the wire-tap channel. 

\subsubsection{Three terminals}
We now consider\footnote{While this manuscript was under review, further work on secrete key generation with polar coding was proposed in \cite{remi}.} the case of three terminals. Let $X^n,Y^n,Z^n$ be i.i.d.\ from a known distribution on $\F_2^3$ and assume that $X^n=Y^n + A^n$ and $Z^n=Y^n + B^n$, where $A^n, B^n$ are independent with, say, $H(A_1) \leq H(B_1)$. Hence there exists $\Delta^n$ independent of $A^n$ such that $B^n=A^n+\Delta^n$. One obtains 
$$H(V_i|V^{i-1},Z^n) = H(V_i|V^{i-1},Y^n+A^n + \Delta^n) \geq H(V_i|V^{i-1},Y^n+A^n) = H(V_i|V^{i-1},X^n).$$
Hence, the terminal having $Y$ sends the components for which $H(V_i|V^{i-1},Z^n) \geq \e$ to the two other terminals. 
Since these components contain the components for which $H(V_i|V^{i-1},X^n) \geq \e$, both terminals can recover $Y$ and generate a secret key of rate $I(Y;Z)=\min(I(X;Y),I(Y;Z))$, which is optimal according to \cite{narayan}.

\section{Proofs}
\subsection{Proof of Theorem \ref{main}}

\begin{lemma}\label{mainlemma}
%For any $\beta<1/2$, and for $\e_n=2^{-n^{\beta}}$, we have,
For any $\e>0$, we have,
\begin{align*}
\frac{1}{n} |\{j \in [n]:  H(Y_j[S]|Y^{j-1}) \in \mZ \pm \e, \forall S \subseteq [m] \}| \\ \to 1 
\end{align*}
%\begin{align*}
%\frac{1}{n} |\{j \in [n]: \,& \exists A_j \text{ with } \rank(A_j)= m-\mathrm{int}(H(Y_j|Y^{j-1})),\\
%& H(A_j Y_j |Y^{j-1}) \leq \e  \}|  \to 1.
%\end{align*}
\end{lemma}
%This lemma implies the first part of Theorem \ref{main}, as shown in next section. 
%The second part of the theorem is proved in Section  \ref{proofs2}, together with the following result, which further characterizes the dependency structure of $Y$.

In order to prove Lemma \ref{mainlemma}, we need the following definition and lemmas.
\begin{definition}
For a random vector $V$ distributed over $\F_2^m$, define $V^{-}=V+V'$ and $V^{+}=V'$, where $V'$ is an i.i.d.\ copy of $V$.
Let $\{b_i\}_{i \geq 1}$ be i.i.d.\ binary random variables uniformly distributed in $\{-,+\}$. 
Denote by $V^{b_1 \dots b_k}$ the iteration of the $+$ and $-$ operations on $V$. Let
\begin{align*}
%& \eta_k=H(V^{b_1 \dots b_k}| V^{c_1 \dots c_k}, \forall (c_1 \dots c_k) < (b_1 \dots b_k))\\
%& \eta_k[S]=H(V^{b_1 \dots b_k}[S]| V^{b_1 \dots b_k}[S^c], V^{c_1 \dots c_k}, \\
%& \phantom{\eta_k[S]=H(V^{b_1 \dots b_k}[S]|} \forall (c_1 \dots c_k) < (b_1 \dots b_k))
& \eta_k[S]=H(V^{b_1 \dots b_k}[S]| V^{c_1 \dots c_k}, \forall (c_1 \dots c_k) < (b_1 \dots b_k))
\end{align*}
for $S \subseteq [m]$, where the order between $(-,+)$-sequences is the lexicographic order with $- < +$ (for example $-+ < +-$).
\end{definition}

Note that $$\{ V^{b_1 \dots b_k} : (b_1 \dots b_k) \in \{-,+\}^k \} \stackrel{(d)}{=} X G_{2^k}$$ where $X$ is the matrix whose columns are i.i.d copies of $V$. The following lemma explains the previous definition.  
\begin{lemma}\label{corresp}
Using $V \sim \mu$ in the definition of $\eta_k[S]$, we have for any $n$ and any set $D \subseteq [0,|S|]$
%$$\frac{1}{n} \{j \in [n] : H(Y_j[S]| Y^{j-1} Y_j[S^c]) \in D \} = \pp\{ \eta_{\log_2 (n)}[S] \in D \}.$$
$$\frac{1}{n} |\{j \in [n] : H(Y_j[S]| Y^{j-1} ) \in D \}| = \pp\{ \eta_{\log_2 (n)}[S] \in D \}.$$
\end{lemma}
The proof is a direct consequence from the fact that the $b_k$'s are i.i.d.\ uniform.
Using the invertibility of $\bigl[\begin{smallmatrix} % or pmatrix or bmatrix or Bmatrix or ...
      1 & 0 \\
      1 & 1 \\
   \end{smallmatrix}\bigr]$ and properties of the conditional entropy, we have the following. 
\begin{lemma}
$\eta_k[S]$ is a super-martingale with respect to $b_k$ for any $S \subseteq [m]$ and a martingale for $S=[m]$.
\end{lemma}
\begin{proof}
%We used the shorthand notation ``past'' in the conditioning of the entropy to include all random variables in the past (as defined above) of the random variable in the left-hand side of the conditional entropy.  
For $n=2$, we have
\begin{align}
2 H(X_1 [S]  ) &=  H(X_1 [S] X_2[S] ]) \notag \\
&= H(Y_1 [S] Y_2[S] ) \notag \\
&= H(Y_1 [S]  ) + H(Y_2 [S]  | Y_1[S]  ) \notag \\
& \geq H(Y_1 [S] ) + H(Y_2 [S]  | Y_1 ) \label{last} 
%&2 H(X_1 [S]  | X_1[S^c]) \notag \\
%&=  H(X_1 [S] X_2[S] | X_1[S^c] X_2[S^c]) \notag \\
%& = H(Y_1 [S] Y_2[S] | Y_1[S^c] Y_2[S^c]) \notag \\
%& = H(Y_1 [S]  | Y_1[S^c] Y_2[S^c]) + H(Y_1 [S]  | Y_1[S^c] Y_2 ) \notag \\
%& \leq H(Y_1 [S]  | Y_1[S^c] ) + H(Y_1 [S]  | Y_1[S^c] Y_2 ) \label{last} 
\end{align}
with equality in \eqref{last} if $S=[m]$.
For $n\geq 2$, the same expansion holds including in the conditioning the appropriate ``past'' random variables.
\end{proof}

Note that because $\eta_k[S]$ is a martingale for $S=[m]$, the sum-rate $H(\mu)$ is conserved through the polarization process. 
Now, using the previous lemma and the fact that $\eta_k[S] \in [0,|S|]$ for any $S$, the martingale convergence theorem implies the following.
\begin{corol}\label{conv}
For any $S \subseteq [m]$, $\eta_k[S]$ converges almost surely.
\end{corol}

The following allows to characterize possible values of the process $\eta_k[S]$ when it converges.

\begin{lemma}\label{invar}
For any $\e>0$, $X$ valued in $\F_2^{m}$, $Z$ arbitrary, $(X',Z')$ an i.i.d.\ copy of $(X,Z)$, $S \subseteq [m]$, there exists $\delta=\delta(\e)$ such that
$H( X'[S] | Z') - H(X' [S]| Z, Z',X[S]+X'[S]) \leq \delta$ implies $H(X'[S] | Z')-H(X'[S \setminus i] | Z')  \in \{0,1 \} \pm \e$ for any $i \in S$.
%$H( X'[S] | Y',X'[S^c]) - H(X' [S]| Y,Y',X[S]+X'[S],X' [S^c]) \leq \delta$ implies $H(X'[S] | Y',X' [S^c])-H(X'[S \setminus i] | Y',X' [(S \setminus i)^c])  \in \{0,1 \} \pm \e$ for any $i \in S$.
\end{lemma}
\begin{proof}
We have
\begin{align}
&H( X'[S] | Z') - H(X' [S]| Z, Z',X[S]+X'[S]) \notag \\
&=I(X'[S]; X[S]+X'[S]| Z ,Z') \notag \\
& \geq I(X'[S]; X[i]+X'[i]| Z, Z', X[S\setminus i]+X'[S\setminus i]) \notag \\
& \geq I(X'[i]; X[i]+X'[i]| Z ,Z', X[S\setminus i]+X'[S\setminus i], X'[S \setminus i]) \notag \\
& = H(X'[i]| Z', X'[S \setminus i]) - H(X'[i]| Z ,Z',  X[S\setminus i],X'[S \setminus i],X[i]+X'[i]) . \label{squiz}
\end{align}
The point of previous expansions is to bring the expression appearing in a one-step polarization transform of a scalar (or single-user) source. 
It is shown in \cite{ari3} that if $A_1,A_2$ are binary random variables and $B_1,B_2$ are valued in arbitrary sets such that 
$\pp_{A_1 A_2 B_1 B_2} (a_1,a_2,b_1,b_2) = P(a_1+a_2)P(a_2) Q(b_1|a_1+ a_2) Q(b_2|a_2)$, for some conditional probability distribution $Q$ and binary probability distribution $P$,
then, for any $a>0$, there exists $b >0$ such that $H(A_2|B_2) - H(A_2|B_1 B_2 A_1) \leq b$ implies $H(A_2|B_2) \in\{0,1\} \pm a$. This is  exactly the setting posed in equation (2) of \cite{ari3} for the polarization of a single source polarization with side information.  
Hence, we set $A_1=X[i]+X'[i]$, $A_2=X'[i]$, $B_1=Z , X[S\setminus i]$ and $B_2=Z' , X'[S\setminus i]$ and we can pick $\delta$ small enough to lower bound \eqref{squiz} and show that $H(X'[i]| Z', X'[S \setminus i]) \in \{0,1\} \pm \e$. 
From the chain rule, we conclude that $H(X'[S]| Z')-H(X'[S \setminus i]| Z') \in \{0,1\} \pm \e$.
\end{proof}
%We refer to Lemma 4 in \cite{mmac} for the proof of previous lemma, since it follows from it (the entropy setting does not lead to a much direct proof). 
We then get the following using Corollary \ref{conv} and Lemma \ref{invar}.
\begin{corol}\label{integer}
With probability one, $\lim_{k \to \infty} \eta_k[S] \in \{0,1,\dots,|S|\}$.
\end{corol}

Finally, Lemma \ref{corresp} and Corollary \ref{integer} imply Lemma \ref{mainlemma}. The next lemma proves the second item of the theorem. 

\begin{lemma}\label{tech1}
For a random vector $Y$ valued in $\F_2^m$, and an arbitrary random vector $Z$, if $$H(Y [S]| Z) \in \{0,1,\dots,|S|\} \pm \e$$ for any $S \subseteq [m]$, we have $$  H(\sum_{i \in S} Y[i]|Z) \in \{0,1\} \pm \delta(\e),$$
with $\delta(\e) \stackrel{\e \to 0}{\to}0$.
\end{lemma}
This lemma is proved in \cite{abbematroid} as Lemma 7, we also provide a proof in Section \ref{lemma5}.  
Denote by $a_S$ the indicator vector of $S$, i.e., $a_S \in \{0,1\}^m$ with $a_S(i)=1$ if and only of $i \in S$, and denote by $a_S Y_j =\sum_{i \in S} Y_j[i]$. Note that, if for two sets $S,T \subseteq [m]$, we have $H(a_S Y_j|Z) \leq \delta(\e)$ and $H(a_T Y_j|Z) \leq \delta(\e)$, then $H(a_S Y_j +a_T Y_j |Z)  \leq 2\delta(\e)$. Moreover $H(a_S Y_j +a_T Y_j |Z)=H(a_{S + T} Y_j |Z)  \leq 2\delta(\e)$, where $S+T$ denotes the exclusive OR of the two sets. Hence the set of indicator subsets of $[m]$ which have low conditional entropy is a subspace of $\F_2^m$. 
To keep notations compact, we can organize these subsets in a matrix of rank at most $m$, i.e., for any $j \in E_n$, there exists a matrix $A_j$ of rank $r_j$, such that 
$$H(A_j Y_j |Y^{j-1}) \leq m \delta(\e) .$$
We now show how we can use this other characterization of the dependencies in $Y$ to conclude a speed convergence result. We first need the following ``single-user'' result, following the approach of \cite{mmac}. 

\begin{lemma}\label{tech2}
For any $\beta<1/2$, $\e \in (0,1)$ and $\e_n=2^{-n^{\beta}}$, we have 
\begin{align*}
\frac{1}{n} | \{j \in [n]: \e_n < H(\sum_{i \in S} Y_j[i]|Y^{j-1}) <\e , \text{ for any } S \subseteq [m] \} |  \to 0 .
\end{align*}
\end{lemma}
\begin{proof}
%Note that $H(\sum_{i \in S} Y_j[i]|Y^{j-1})$
We define the auxiliary family of random processes $\zeta_k[S]$, for $S \subseteq [m]$, by
%\begin{align*}
%%& \eta_k=H(V^{b_1 \dots b_k}| V^{c_1 \dots c_k}, \forall (c_1 \dots c_k) < (b_1 \dots b_k))\\
%%& \eta_k[S]=H(V^{b_1 \dots b_k}[S]| V^{b_1 \dots b_k}[S^c], V^{c_1 \dots c_k}, \\
%%& \phantom{\eta_k[S]=H(V^{b_1 \dots b_k}[S]|} \forall (c_1 \dots c_k) < (b_1 \dots b_k))
%& \nu_k[S]=H(\sum_{i \in S}V^{b_1 \dots b_k}[i]| V^{c_1 \dots c_k}, \forall (c_1 \dots c_k) < (b_1 \dots b_k))
%\end{align*}
\begin{align*}
%& \eta_k=H(V^{b_1 \dots b_k}| V^{c_1 \dots c_k}, \forall (c_1 \dots c_k) < (b_1 \dots b_k))\\
%& \eta_k[S]=H(V^{b_1 \dots b_k}[S]| V^{b_1 \dots b_k}[S^c], V^{c_1 \dots c_k}, \\
%& \phantom{\eta_k[S]=H(V^{b_1 \dots b_k}[S]|} \forall (c_1 \dots c_k) < (b_1 \dots b_k))
& \zeta_k[S]=Z(\sum_{i \in S}V^{b_1 \dots b_k}[i]| V^{c_1 \dots c_k}, \forall (c_1 \dots c_k) < (b_1 \dots b_k))
\end{align*}
where, for a binary uniform random variable $A$ and an arbitrary random variable $B$, $Z(A|B)=2 \E_B (\pp\{ A=0 | B \} \pp\{ A=1 | B \})^{1/2}$ is the Bhattacharyya parameter. Note that
\begin{align}
Z(A|B) \geq H(A|B). \label{bata}
\end{align}
This follows from Proposition 2 in \cite{ari3}.
We then have, using the chain rule and source polarization inequalities on the Bhattacharyya parameter, namely Proposition 1 in \cite{ari3}, that 
\begin{align*}
& \zeta_{k+1}[S] \leq \zeta_{k}[S]^2 \text{ if } b_{k+1}=1,\\% \label{z1} \\
& \zeta_{k+1}[S] \leq 2 \zeta_k[S]  \text{ if } b_{k+1}=0, %\label{z2}
\end{align*}
and using Theorem 3 of \cite{ari2}, we conclude that for any $\alpha < 1/2$
$$\lim \inf_{\ell \rightarrow \infty} \pp(\zeta_k \leq 2^{-2^{\alpha k}}) \geq \pp( \zeta_\infty=0).$$
Finally, the conclusion follows from \eqref{bata}. 
\end{proof}

Finally, we have
\begin{align}
\frac{1}{n} | \{j \in [n]: 
&H(Y_j [S]| Y^{j-1}) \in \{0,1,\dots,|S|\} \pm \e, \forall S \subseteq [m], \notag \\
&\exists A_j \text{ with } \rank (A_j)= \mathrm{int}( m-H(Y_j | Y^{j-1})),\notag \\
&H(A_jY_j |Y^{j-1}) <\e_n  \} | \to 1 . \label{set}
\end{align}
This holds since the property that $H(Y_j [S]| Y^{j-1}) \in \{0,1,\dots,|S|\} \pm \e$, for all $S \subseteq [m]$, together with Lemma \ref{tech1}, imply that $H(\sum_{i \in S} Y_j[i]|Y^{j-1}) \in \{0,1\} \pm \delta(\e)$ with $\delta(\e) \to 0$ as $\e \to 0$. 
%Consider now the sets $S$ for which $H(\sum_{i \in S} Y_j[i]|Y^{j-1}) < \delta$. 
Let $\delta>0$ sufficiently small, and assume that $H(\sum_{i \in S} Y_j[i]|Y^{j-1}) \in \{0,1\} \pm \delta$. 
Note that $H(\sum_{i \in S_1} Y_j[i]|Y^{j-1}) < \delta$ and $H(\sum_{i \in S_2} Y_j[i]|Y^{j-1}) < \delta$ imply $H(\sum_{i \in S_1 + S_2} Y_j[i]|Y^{j-1}) < 2 \delta$, where $S_1+S_2$ is the disjoint union (or symmetric difference) of $S_1$ and $S_2$. Take now all subsets of $[m]$ for which $H(\sum_{i \in S} Y_j[i]|Y^{j-1}) < \delta$, and create a matrix whose rows are the indicator vectors of these subsets, and remove redundant rows to obtain an equivalent matrix $A_j$ which must have at most $m$ rows. Then $H(A_j Y_j|Y^{j-1}) < m \delta$ (which is small since $m$ is fixed). Moreover, the rank of $A_j$, i.e., the number of non-redundant rows is determined as follows. Consider all subsets $T$ such that $H(\sum_{i \in T} Y_j[i]|Y^{j-1}) > 1- \delta$. Take now a basis for these subsets (i.e., as before, a set of non-redundant indicator vectors for these subsets). The cardinality of this basis must be the closet integer to $H(Y_j| Y^{j-1})$ (recall that $H(Y_j| Y^{j-1})$ is close to an integer), since $H(\sum_{i \in S} Y_j[i]|Y^{j-1}) \in \{0,1\} \pm \delta$ for all $S \subseteq [m]$. Hence the rank of $A_j$ must be the complement of the closet integer to $H(Y_j| Y^{j-1})$, i.e., $\rank (A_j)= \mathrm{int}( m-H(Y_j | Y^{j-1}))$, where $\mathrm{int}(x)$ is the closest integer to $x$.  
%Arranging those sums which are vanishing in a matrix $A_j$ of maximal rank as explained previously, we must have by the rank theorem that the rank of $A_j$ is equal to the dimension $m$ minus the integral rounding\footnote{For a real number $x$, we define $\mathrm{int}(x)$ as the closest integer to $x$.} of the total entropy $H(Y_j | Y^{j-1})$, where the integral rounding is justified by the fact that this conditional entropy must be close to an integer. 
Finally, we need to ensure that the sums which are vanishing are indeed vanishing fast, i.e., $H(A_jY_j |Y^{j-1}) <\e_n$, but this follows directly from Lemma \ref{tech2}.

\subsection{Proof of Theorem \ref{sw}}

For $j \in [n]$, the components to be decoded in $Y_j$ are not correctly decoded with probability
$$P_e(j) \leq H(A_jY_j |Y^{j-1}) \leq \e_n.$$
Hence the block error probability is bounded as
$$P_e \leq \sum_{j \in [n]} P_e(j) \leq n \e_n,$$
and we obtain a block error probability of $O(2^{-n^{\beta}})$ for any $\beta < 1/2$.

Note that for a fix $m$, the complexity is $O(n \log n)$ by applying the same divide-and-conquer factorization of the likelihoods as developed in \cite{ari}, to our setting where the + and - operations are over $GF(2^m)$ instead of $GF(2)$. 

The achievability of the sum rate is ensured by the second item of Theorem \ref{main}. We now show the last part of the theorem.   
%\begin{lemma}
%Let $X^n=(X_1,\dots,X_n)^T$ be i.i.d.\ random variables distributed under $\mu$ on $\F_p^m$, with $\mu$ separable in $S$, 
Assume  that $\mu$ separable in $S$, then
%and let $Y^n=G_n X^n$. % $\mathcal{H}:=\mathcal{H}(X_1)$ and $\mathcal{H}_i:=\mathcal{H}(Y_i | Y{i-1})$.
%We show that 
\begin{align}
n H(X_1[S]|X_1[S^c]) = \sum_{i=1}^n H(Y_i[S]|Y_i[S^c], Y^{i-1}) .
\end{align}
%\begin{align}
%n \mathcal{H}(X_1) [S] = \sum_{i=1}^n \mathcal{H}(Y_i|Y^{i-1})[S] .
%\end{align}
%\begin{align}
%\mathcal{H} [S] = \sum_{i=1}^n \mathcal{H}_i [S] .
%\end{align}
%\end{lemma}

%\noindent
%{\bf Remark.} With the result on joint polarization for correlated sources, this lemma implies that if $\mu$ is separable in $S$, then polarization allows to achieve full sum-rate for the users in $S$. In particular, if $\mu$ is separable in every subsets of the users, then polarization achieves the full rate region. 
 
%\begin{proof}

To see this, assume w.l.o.g.\ that $S=\{1,\dots,k\}$ for $k \leq m$. 
%To simplify notations, we use in this proof column vector representations. 
Hence, $X^n=[X^n[S], X^n[S^c]]$.
By the separability assumption, there exist $W_1[S], \dots, W_n[S]$ i.i.d.\ and $F$ such that 
\begin{align}
X^n=[F(X^n[S^c])+W^n[S], X^n[S^c]], 
\end{align}
where $F$ is applied component-wise. 
Hence
\begin{align}
Y^n=X^n G_n&=[F(X^n[S^c])G_n+Z^n[S], Y^n[S^c]] \\
&=[F(Y^n[S^c])+Z^n[S], Y^n[S^c]]
\end{align}
where $Z^n=W^n G_n$ and where we used the fact $F$ is linear in the last equality. 
Therefore, 
\begin{align}
 H(Y_i[S]|Y_i[S^c], Y^{i-1}) &= H(F(Y_i[S^c])+Z_i[S] | Y_i[S^c], Y^{i-1})\\
&=H(Z_i[S] | Y_i[S^c], Z^{i-1}[S], Y^{i-1}[S^c])\\
%&=H(Z_i[S] | Y_i[S^c], Z^{i-1}[S], Y^{i-1}[S^c])\\
&=H(Z_i[S] | Z^{i-1}[S])
\end{align}
and
\begin{align}
\sum_{i=1}^n  H(Y_i[S]|Y_i[S^c], Y^{i-1}) &=\sum_{i=1}^n H(Z_i[S] | Z^{i-1}[S])\\
&= H(Z^n[S])\\
%&=H(W^n[S])\\
&=n H(W_1[S]).
\end{align}
Moreover,
\begin{align}
H(X_1[S]| X_1[S^c]) &= H(F(X_1[S^c])+W_1[S] | X_1[S^c]) =H(W_1[S] ).
\end{align}
%\end{proof}

%Note if $X[S]$ is expressed as $MZ$ where $M$ is an $|S| \times m$ matrix and $Z$ is a $m \times 1$ random vector with independent components. 

\subsection{Proof of Theorem \ref{sk}}

\begin{proof}[Proof of Theorem \ref{sk} for the polar-key-1 protocol]
Let $\e = 2^{-n^\beta}$, $(X^n,Y^n) \iid \mu$ and let $U^n= X^nG_n$, $V^n= Y^nG_n$. 
Let $\hat{U}^n$ be the output of the decoding algorithm in \cite{ari3} run by Bob in possession of $U^n[R_{\e,n}(X|Y)]$ and $Y^n$.
From Theorem 3 in \cite{ari3}, $\pp\{ \hat{U}^n \neq U^n \} =  o(2^{-n^\beta})$. Hence the first claim in the theorem.  
The only data communicated publicly in the protocol is $C_n:=U^n[R_{\e,n}(X|Y)]$, and under correct decoding, the key is 
\begin{align}
S_n&=U^n[R_{1-\e,n}(X) \setminus R_{\e,n}(X|Y)]\\
&=U^n[R_{1-\e,n}(X) \setminus (R_{1-\e,n}(X|Y) \cup  \Delta_{\e,n}(X|Y))]
\end{align}
where
\begin{align}
R_{\e,n}(X) &:=\{i \in [n] : H(U_i| U^{i-1})  \geq \e \},\\
R_{\e,n}(X|Y) &:=\{i \in [n] : H(U_i| U^{i-1},V^n)  \geq \e \},\\
\Delta_{\e,n}(X|Y) &:= \{i \in [n] : H(U_i| U^{i-1},V^n)  \in (\e,1-\e) \}  . 
\end{align}
Therefore
\begin{align}
I(C_n;S_n)&= I(U^n[R_{1-\e,n}(X|Y) \cup \Delta_{\e,n}(X|Y)];U^n[R_{1-\e,n}(X) \setminus (R_{1-\e,n}(X|Y) \cup \Delta_{\e,n}(X|Y)]))\\
&\leq I(U^n[R_{1-\e,n}(X|Y) \cup \Delta_{\e,n}(X|Y)];U^n[R_{1-\e,n}(X) \setminus R_{1-\e,n}(X|Y)])\\
&\leq  a(n) + b(n)
\end{align}
where
\begin{align}
a(n)&:=I(U^n[R_{1-\e,n}(X|Y) ];U^n[R_{1-\e,n}(X) \setminus R_{1-\e,n}(X|Y)]) \\
b(n)&:=I( U^n[\Delta_{\e,n}(X|Y)] ; U^n[R_{1-\e,n}(X)  | U^n[R_{1-\e,n}(X|Y) ]) .
\end{align}
Note that by the source polarization \cite{ari3}, 
\begin{align}
b(n) \leq | \Delta_{\e,n}(X|Y) | = o(n).
\end{align}
Moreover, since $R_{1-\e,n}(X|Y) \subseteq R_{1-\e,n}(X)$, 
\begin{align}
a(n)
%&=H(U^n[R_{\e,n}(X|Y)]) -H(U^n[R_{\e,n}(X|Y)]|U^n[R_{\e,n}(X) \setminus R_{\e,n}(X|Y)])
&=H ( U^n[R_{1-\e,n}(X) \setminus R_{1-\e,n}(X|Y)])+H(U^n[R_{1-\e,n}(X|Y)]) -H(U^n[R_{1-\e,n}(X)])  \label{eq1}\\
&\leq |R_{1-\e,n}(X)| -  H(U^n[R_{1-\e,n}(X)])
\end{align}
%one obtains
%\begin{align}
%I(S_n;C_n)&=H(U^n[R_{\e,n}(X|Y)]) -H(U^n[R_{\e,n}(X)]) \notag \\& + H ( U^n[R_{\e,n}(X)]\setminus U^n[R_{\e,n}(X|Y)]) \notag  \\
%& \leq |R_{\e,n}(X)| -  H(U^n[R_{\e,n}(X)]). \label{eq1}
%\end{align}
and 
%\begin{align*}
%&nH(X_1)=H(X^n)=H(U^n) \\ 
%&=H(U^n[R_{\e,n}(X)]) + H(U^n[R_{\e,n}^c(X)] | U^n[R_{\e,n}(X)]) 
%\end{align*}
%and
\begin{align}
H(U^n[R_{1-\e,n}(X)]) \geq \sum_{i \in R_{1-\e,n}(X)} H(U_i|U^{i-1}) \geq (1-\e)|R_{1-\e,n}(X)|. \label{eq2} 
%&= (1-\e) (nH(X_1)+o(n))  
\end{align}
%which means that $U^n[R_{\e,n}(X)]$ is roughly uniformly distributed over $R_{\e,n}(X)$. 
Hence, from \eqref{eq1} and \eqref{eq2},   
%\begin{align*}
%&H(U^n[R_{\e,n}(X|Y)])\\
% &= H(U^n[R_{\e,n}(X)]) - H(U^n[R_{\e,n}(X) \setminus R_{\e,n}(X|Y)] | U^n[R_{\e,n}(X|Y)]) \\
%& \geq (1-\e)|R_{\e,n}(X)| - (|R_{\e,n}(X)|-|R_{\e,n}(X|Y)|) \\
%& =  |R_{\e,n}(X|Y)|-\e|R_{\e,n}(X)| 
%\end{align*}
\begin{align*}
a(n) \leq \e |R_{\e,n}(X)| \leq \e n. 
\end{align*}
Finally, 
\begin{align}
H(U^n[R_{1-\e,n}(X) \setminus R_{\e,n}(X|Y) ] ) &\geq H(U^n[R_{1-\e,n}(X)]) - H(U^n[R_{\e,n}(X|Y) ] ) \\
%&\geq  H(U^n[R_{\e,n}(X)]) - |R_{\e,n}(X|Y)| \notag \\
& \geq  (1-\e)|R_{1-\e,n}(X) | - |R_{\e,n}(X|Y)| \notag \\
& = (1-\e) (nH(X_1)+o(n)) - (nH(X_1|Y_1)+o(n)) \label{pol} \\
& = n I(\mu) + o(n) \notag.
\end{align}
%\begin{align}
%&H(U^n[R_{\e,n}(X)\setminus R_{\e,n}(X|Y)]) \notag \\
%& = H(U^n[R_{\e,n}(X) - H(U^n[R_{\e,n}(X|Y)]) \notag \\
%& \geq  H(U^n[R_{\e,n}(X)]) - |R_{\e,n}(X|Y)| \notag \\
%& \geq  (1-\e)|R_{\e,n}(X) | - |R_{\e,n}(X|Y)| \notag \\
%& = (1-\e) (nH(X_1)+o(n)) - (nH(X_1|Y_1)+o(n)) \label{pol} \\
%& = n I(X_1;Y_1) + o(n) \notag
%\end{align}
%where \eqref{pol} follows from \cite{ari3}. 
Since the secret-key cannot have entropy rate more than $I(\mu)$, the conclusion follows.  
\end{proof}

\begin{proof}[Proof of Theorem \ref{sk} for the polar-key-2 protocol]
Let $\e = 2^{-n^\beta}$, $(X^n,Y^n) \iid \mu$ and let $U^n= X^nG_n$, $V^n= Y^nG_n$. 
Let $\hat{U}^n$ be the output of the {\it polar-matrix-decoder} run by Bob in possession of $U^n[Q_{\e,n}(X|Y)]$ and $Y^n$.
From Theorem \ref{sw}, we have $\pp\{ \hat{U}^n \neq U^n \} =  o(2^{-n^\beta})$. Hence the first claim in the theorem.  
The only data communicated publicly in the protocol is $C_n:=U^n[Q_{\e,n}(X|Y)]$, and under correct decoding, the key is 
\begin{align}
S_n&=U^n[R_{1-\e,n}(X) \setminus Q_{\e,n}(X|Y)]\\
&=U^n[R_{1-\e,n}(X) \setminus (\widetilde{R}_{1-\e,n}(X|Y) \cup  \Omega_{\e,n}(X|Y))]
\end{align}
where
\begin{align}
R_{\e,n}(X) &:=\{i \in [n] : H(U_i| U^{i-1})  \geq \e \},\\
\widetilde{R}_{1-\e,n}(X|Y) &:= \{i \in [n] : H_i  \in \{(1,0,1),(1,1,2)\} \pm \e \},\\
\Omega_{\e,n}(X|Y) &:= \{i \in [n] : H_i  \notin \mZ^3 \pm \e \} ,\\
Q_{\e,n}(X|Y) &:=  \Omega_{\e,n}(X|Y) \cup \widetilde{R}_{1-\e,n}(X|Y).
%&= \{i \in [n] : H(U_i | U^{i-1} V^{i-1} V_i)  \in  (\e,1-\e) \} \cup \{i \in [n] : H(U_i | U^{i-1} V^{i-1} V_i) \geq 1-\e \}\\
%&=T_{1-\e,n}(X|Y) \cup \Delta_{\e,n}(X|Y) \\
%\Delta_{\e,n}(X|Y) &:= \{i \in [n] : H_i  \notin \mZ \pm \e \} \\
%T_{1-\e,n}(X|Y) &:=\{i \in [n] : H(U_i | U^{i-1} V^{i-1} V_i) \geq 1-\e \}
\end{align}

Therefore
\begin{align}
I(C_n;S_n)&= I(U^n[\widetilde{R}_{1-\e,n}(X|Y) \cup \Omega_{\e,n}(X|Y)];U^n[R_{1-\e,n}(X) \setminus (\widetilde{R}_{1-\e,n}(X|Y) \cup  \Omega_{\e,n}(X|Y)) ] )\\
&\leq I(U^n[\widetilde{R}_{1-\e,n}(X|Y) \cup \Omega_{\e,n}(X|Y)];U^n[R_{1-\e,n}(X) \setminus \widetilde{R}_{1-\e,n}(X|Y) ] )\\
&\leq  \alpha(n) + \beta(n)
\end{align}
where
\begin{align}
\alpha(n)&:=I(U^n[\widetilde{R}_{1-\e,n}(X|Y) ];U^n[R_{1-\e,n}(X) \setminus \widetilde{R}_{1-\e,n}(X|Y)]  ) \\
\beta(n)&:=I(\Omega_{\e,n}(X|Y);U^n[R_{1-\e,n}(X) ]  | U^n[\widetilde{R}_{1-\e,n}(X|Y) ]) .
\end{align}
By Theorem \ref{main}, 
\begin{align}
\beta(n) \leq | \Omega_{\e,n}(X|Y) | = o(n).
\end{align}
Moreover, since $\widetilde{R}_{\e,n}(X|Y) \subseteq R_{\e,n}(X)$, 
\begin{align}
\alpha(n)
%&=H(U^n[R_{\e,n}(X|Y)]) -H(U^n[R_{\e,n}(X|Y)]|U^n[R_{\e,n}(X) \setminus R_{\e,n}(X|Y)])
&=H(U^n[R_{1-\e,n}(X) \setminus \widetilde{R}_{1-\e,n}(X|Y)]) - H( U^n[R_{1-\e,n}(X) \setminus \widetilde{R}_{1-\e,n}(X|Y)]  | U^n[\widetilde{R}_{1-\e,n}(X|Y) ])  \\
& \leq H(U^n[R_{1-\e,n}(X) \setminus \widetilde{R}_{1-\e,n}(X|Y)]) - H( U^n[R_{1-\e,n}(X)) + H( U^n[\widetilde{R}_{1-\e,n}(X|Y) ])\\ 
&\leq |R_{1-\e,n}(X)| -  H(U^n[R_{1-\e,n}(X)])\\
& \leq \e n.
\end{align}
Finally 
\begin{align}
H(U^n[R_{1-\e,n}(X) \setminus Q_{\e,n}(X|Y) ] ) &\geq H(U^n[R_{1-\e,n}(X)]) - H(U^n[Q_{\e,n}(X|Y) ] ) \\
%&\geq  H(U^n[R_{\e,n}(X)]) - |R_{\e,n}(X|Y)| \notag \\
& \geq  (1-\e)|R_{1-\e,n}(X) | - |Q_{\e,n}(X|Y)| \notag \\
%& = (1-\e) (nH(X_1)+o(n)) - (nH(X_1|Y_1)+o(n)) \label{pol} \\
& \geq (1-\e) (nH(X_1)+o(n))- (n \frac{1}{2}(H(X_1)+H(X_1|Y_1))+o(n)) \label{pol2} \\
& \geq n \frac{1}{2}I(\mu) + o(n) \notag.
\end{align}
%\begin{align}
%&H(U^n[R_{\e,n}(X)\setminus R_{\e,n}(X|Y)]) \notag \\
%& = H(U^n[R_{\e,n}(X) - H(U^n[R_{\e,n}(X|Y)]) \notag \\
%& \geq  H(U^n[R_{\e,n}(X)]) - |R_{\e,n}(X|Y)| \notag \\
%& \geq  (1-\e)|R_{\e,n}(X) | - |R_{\e,n}(X|Y)| \notag \\
%& = (1-\e) (nH(X_1)+o(n)) - (nH(X_1|Y_1)+o(n)) \label{pol} \\
%& = n I(X_1;Y_1) + o(n) \notag
%\end{align}
%where \eqref{pol} follows from \cite{ari3}. 
where \eqref{pol2} follows from the fact that a rate on the dominant face is achieved by Theorem \ref{main}, hence, the gap between $H(X_1|Y_1)$ and the achieved point on the dominant face is at least $\frac{1}{2}(H(X_1)+H(X_1|Y_1))$, since we chose appropriately the role of Alice and Bob to maximize this gap.
\end{proof}

\section*{Acknowledgement}
I would like to thank Ido Tal and Alexander Vardy for a careful reading of a first version of this manuscript \cite{corr} in 2011. I would also like to thank Ueli Maurer for stimulating discussions on secret key agreement.

\section{Appendix A: proof of Lemma \ref{comp-lemma}}\label{proofA}
Let $X,Y$ be two discrete random variables supported on $\X,\Y$ respectively. The lemma applies to $Y=f(X)$, but we write the proof for the more general setting. Upon observing $Y=y$, the estimator of $X$ which minizimes the probability of a wrong estimate is $\hat{x}(y)= \arg\max_{x \in \X} \Pr\{X=x|Y=y\}$ (i.e., MAP decoding), and the error probability is %$1-\max_{x \in \X} \Pr\{X=x|Y=y\}$. 
\begin{align}
P_e(X|Y=y):=1- \max_{x \in \X} \Pr\{X=x|Y=y\} \Pr\{Y=y\} = 1- e^{-H_\infty (X|Y=y)}.
\end{align}
Averaging over the realizations of $Y$, and using Jensen's inequality, yields an average error probability of
\begin{align}
P_e(X|Y):=1-\sum_{y \in \Y} \max_{x \in \X} \Pr\{X=x|Y=y\} \Pr\{Y=y\} \leq 1- e^{-H_\infty (X|Y)}.
\end{align}
Note that for $Y=f(X)$, where $f: \F_2^n \to \F_2^m$, hence the above can be simplified to 
\begin{align}
P_e(X|f(X))=1-\sum_{y \in \F_2^m} \,\, \max_{x \in f^{-1}(y)} \Pr(X=x) \sum_{u \in f^{-1}(y)} \Pr(X=u). 
\end{align}
%Assume now that $Y=f(X)$. 
If $H_\infty (X|Y) \leq \e$, then clearly $P_e(X|Y)\leq 1 - e^{-\e} = \e + o(\e)$, using MAP as the algorithm. Conversely, if there exists an algorithm that has average error probability less than $\e$, then MAP has average error probability less than $\e$ as well, i.e., 
\begin{align}
P_e(X|Y) = \sum_{y \in \Y} (1-e^{-H_\infty (X|Y=y)}) \Pr\{Y=y\} \leq \e.
\end{align}
%or equivalently, 
%\begin{align}
%%1-\e + o(\e) \leq  \sum_{y \in \Y} e^{-H_\infty (X|Y=y)} \Pr\{Y=y\} \leq 1\\
%0 \leq  \sum_{y \in \Y} (1-e^{-H_\infty (X|Y=y)}) \Pr\{Y=y\} \leq \e.
%\end{align}
Defining $F(y)=H_\infty (X|Y=y)$ and letting $\tau>0$, Markov's inequality implies 
\begin{align}
\Pr\{1-e^{-F(Y)} \geq \e/\tau\} \leq \tau 
\end{align}
hence, since $H_\infty(X|Y=y) \leq H_0(X|Y=y)=m$, 
\begin{align}
H_\infty (X|Y)= \E F(Y) &= \sum_{y : F(y) < 1-e^{-\e/\tau}} F(y)  \Pr\{Y=y\} + \sum_{y : F(y) \geq 1-e^{-\e/\tau}} F(y)  \Pr\{Y=y\}\\
&\leq 1-e^{-\e/\tau} + m \tau \\
&= \e/\tau + o(\e/\tau) + m \tau.
\end{align}
The conclusion is obtained by choosing $\tau = \sqrt{\e/m}$.

\section{Appendix B: proof of Lemma \ref{tech1}}\label{lemma5}
Let $m \geq 1$, and let $Z_1,\dots,Z_m$ be $m$ random variables taking each value in $\F_2$ and such that $H(Z[S]) \in \{0,1,\dots,|S|\}$ for all $S \subseteq [m]$. %In particular, $H(Z_1),H(Z_2) \in \{0,1\}$ and $H(Z_1,Z_2) \in \{0,1,2\}$. 
Note that an assignment of $(H(Z_1),H(Z_2),H(Z_1,Z_2))$ in $\{0,1\} \times \{0,1\} \times \{0,1,2\}$ determines an assignment of $H(X_1+X_2)$ in $\{0,1\}$. In fact, assuming $H(Z_1)=H(Z_2)=1$ (otherwise it is trivial), it must be that $H(Z_1,Z_2)$ is either 1 or 2. In the first case, $H(Z_1|Z_2)=0$ and $Z_1$ is either equal to $Z_2$ or its flip $(Z_2+1)$, i.e., $H(Z_1+Z_2)=0$. In the second case, $Z_1$ and $Z_2$ are independent and $H(Z_1+Z_2)=1$. This shows that the entropy of the sum of any two random variables is in $\{0,1\}$. Consider now $W_1=Z_1+Z_2$, $W_2=Z_3$. Note that $H(W_1,W_2)=H(Z_1+Z_2,Z_3)$ takes values in $\{0,1,2\}$. In fact,
since $H(Z_1|Z_3) \in \{0,1\}$, $Z_1$ is either a function of $Z_3$ or is independent of $Z_3$, and similarly for $Z_2$ and $Z_3$ since $H(Z_2|Z_3) \in \{0,1\}$. Hence, it must be that $H(Z_1+Z_2|Z_3)$ is either $H(Z_1+Z_2)$, $H(Z_1)$ or $H(Z_2)$, which are all in $\{0,1\}$, and $H(Z_1+Z_2,Z_3)=H(Z_1+Z_2|Z_3) + H(Z_3)$ takes value in $\{0,1,2\}$. Therefore, we can apply the same argument as above for $Z_1,Z_2$ to $W_1,W_2$ and $H(W_1+W_2)=H(Z_1+Z_2+Z_3) \in \{0,1\}$. Similarly, we get that the entropy of any sum of the $Z_i$'s is in $\{0,1\}$. 
 
By the same argument, one obtains the result with a variable in the conditioning, i.e., if $X_1,\dots,X_m$ are $m$ random variables taking each values in $\F_2$ and if $Y$ is a discrete random variable such that $H(X[S]|Y) \in \{0,1,\dots,|S|\}$ for all $S \subseteq [m]$, then $H(\sum_{i \in S} X_i|Y) \in \{0,1\}$ for all $S \subseteq [m]$. For example, the first argument in previous paragraph is as follows. If $H(X_1|Y)=H(X_2|Y)=1$ and $H(X_1,X_2|Y)=1$, then $H(X_1|X_2,Y)=0$ and $X_1$ is a deterministic function of $(X_2,Y)$, or equivalently, $X_1$ and $X_2$ can be determined form each other given $Y$. In other words, given $Y$, either $X_1=X_2$ or $X_1=1+X_2$, hence $H(X_1+X_2|Y)=0$. If instead $H(X_1|Y)=H(X_1|Y)=1$ and $H(X_1,X_2|Y)=2$, then $X_1,X_2,Y$ are mutually independent and each marginally uniform, hence $H(X_1+X_2|Y)=1$. 

Finally, $H(X[S]|Y) \in \{0,1,\dots,|S|\} \pm \e$ implies $H(\sum_{i \in S} X_i|Y) \in \{0,1\} \pm \delta(\e)$ with $\delta(\e) \to 0$ as $\e \to 0$, follows by a continuity argument and the fact that $m$ is fixed.

\end{document}